\documentclass{ws-ijmpb}

\usepackage{hyperref}
\usepackage{amsmath}
\usepackage{amssymb}

\begin{document}

\markboth{D. Georgiev}
{Quantum Zeno effect in the brain}

%
\catchline{}{}{}{}{}
%

\title{MONTE CARLO SIMULATION OF QUANTUM ZENO EFFECT IN THE BRAIN}

\author{DANKO GEORGIEV}

\address{Department of Environmental and Occupational Health, Graduate School of Public Health, University
of Pittsburgh, Bridgeside Point, 100 Technology Drive, Suite \#513\\
Pittsburgh, PA 15219, USA\\
ddg17@pitt.edu}

\maketitle

\begin{history}
\end{history}

\begin{abstract}
Environmental decoherence appears to be the biggest obstacle for successful construction of quantum mind theories. Nevertheless, the quantum physicist Henry Stapp promoted the view that the mind could utilize quantum Zeno effect to influence brain dynamics and that the efficacy of such mental efforts would not be undermined by environmental decoherence of the brain. To address the physical plausibility of Stapp's claim, we modeled the brain using quantum tunneling of an electron in a multiple-well structure such as the voltage sensor in neuronal ion channels and performed Monte Carlo simulations of quantum Zeno effect exerted by the mind upon the brain in the presence or absence of environmental decoherence. The simulations unambiguously showed that the quantum Zeno effect breaks down for timescales greater than the brain decoherence time. To generalize the Monte Carlo simulation results for any $n$-level quantum system, we further analyzed the change of brain entropy due to the mind probing actions and proved a theorem according to which local projections cannot decrease the von Neumann entropy of the unconditional brain density matrix. The latter theorem establishes that Stapp's model is physically implausible but leaves a door open for future development of quantum mind theories provided the brain has a decoherence-free subspace.\\ [5 pt]
Appeared in \emph{Int. J. Mod. Phys. B} 2014; DOI: \href{http://dx.doi.org/10.1142/S0217979215500393}{\path{10.1142/S0217979215500393}}
\end{abstract}



\section{Introduction}

The mainstream view in cognitive neurosciences identifies mind states
with physical states realized within the brain \cite{Borst1982,Churchland1989,Churchland2002}.
Major evidence for such mind-brain identity thesis comes from
the clinical examination of patients with brain trauma in which loss
of certain cognitive abilities occurs \cite{Ramachandran} and from
the ability to elicit subjective experiences by direct electric stimulation
of the brain cortex \cite{Dobelle1973,Dobelle1974,Dobelle1976,Dobelle1979,Dobelle2000}.
If, however, one further postulates that the brain states obey the
deterministic laws of classical mechanics, several counterintuitive
results would follow. For example, the intuitively evident propositions
that we have a \emph{free will} allowing us to make choices, or that
our subjective experiences can have a \emph{causal influence} upon
the brain and the surrounding physical world, would turn out to be
nothing but illusions \cite{Wegner2002}. Defending such a viewpoint
seems to be possible, because the physical reality may not conform
to our expectations of what the physical reality should be. Nevertheless,
an increasing number of scientists think that we could construct a
better theory of mind if we take into account quantum physics \cite{Schrodinger1967,Georgiev2006,Georgiev2013,Beck1992,Jibu1995,Penrose2005}.

One of the most elaborate proposals for a quantum theory of mind is
due to Henry Stapp, who suggested that (1) the mind could influence
the dynamics of the brain using \emph{quantum Zeno effect} \cite{Stapp2001,Stapp2005,Schwartz2005,Stapp2007,Stapp2007b},
and (2) the efficacy of the mental efforts would not be undermined by environmentally
induced decoherence of the brain:
\begin{quote}
The quantum Zeno effect is itself a decoherence effect, and it is
not diminished by environmental decoherence. \emph{Thus the decoherence
argument against using quantum mechanics to explain the influence
of conscious thought upon brain activity is nullified}. \cite[pp.101-102]{Stapp2007b}
\end{quote}
If correct, Stapp's model could provide a scientific basis for the
existence of free will \cite{Stapp2001}. Also, it would explain
how mental efforts causally affect and restructure the organization
of the brain in health or psychiatric disease \cite{Schwartz2005}.
Finally, it would establish that environmental decoherence is not
an obstacle for the construction of quantum theories of mind \cite{Stapp2007,Stapp2007b}.

In a previous work \cite{Georgiev2012}, we have argued that if Stapp's
model worked in the presence of environmental decoherence, then mind
efforts would have been capable of exerting paranormal effects upon
nearby physical measuring devices. Our argument was based on a theorem
stating that if mind efforts operate only upon the brain density matrix
$\hat{\rho}$ using projection operators, then the von Neumann entropy
production cannot be negative \cite{Georgiev2012}. Stapp agreed that
the proof of the particular theorem is correct \cite{Stapp2012},
but argued that it presents no harm for his model because: (1) the
studied two-level model system (polarization of a photon) \cite{Georgiev2012}
is too simple to represent a human brain; (2) the studied quantum
Zeno effect \cite{Georgiev2012} was based on no collapse version
of quantum mechanics, which lacks an essential ingredient of Stapp's
model, namely the wavefunction collapse following the mind probing
action; and (3) the mind action onto the brain density matrix does
not need to slow down the environmental decoherence
in order to exert quantum Zeno effect upon the brain \cite{Stapp2012}.

In this work, to address the physical plausibility of Stapp's model, we performed
computer simulations of quantum Zeno effect exerted by the mind upon
a model quantum brain in the presence or absence of environmental
decoherence. The simulations were meticulously constructed according
to the postulates in Stapp's model (\S \ref{sec:2}). The biological
implementation was based on detailed molecular structural data of
voltage-gated ion channel function in brain cortical neurons (\S \ref{sec:3}).
To address Stapp's claim that the wavefunction collapse is a necessary
ingredient for the model to work, we performed Monte Carlo simulations
in the absence (\S \ref{sec:4}) or presence (\S \ref{sec:5})
of environmental decoherence. The results from the Monte Carlo simulations
unambiguously show that the quantum Zeno effect breaks down and the
brain behaves as a `random telegraph' for timescales greater than
the decoherence time of the brain. The point at which the breakdown
of quantum Zeno effect occurs is assessed by the statistical average
outcome of many such Monte Carlo simulation trials. Because the statistical
average is indistinguishable from the brain dynamics simulated within
no collapse version of quantum mechanics (see \S \ref{sec:4} and \S \ref{sec:5})
one can study the quantum Zeno effect in the brain using unconditional
density matrices only, as done in Ref.~\refcite{Georgiev2012}, without the
necessity to explicitly consider collapses and conditional density
matrices. In \S \ref{sec:6}, we prove a theorem according to which
operating only locally on the brain with projection operators cannot
decrease the von Neumann entropy of the unconditional brain density
matrix. This allows us to generalize the results from the Monte Carlo
simulations to any $n$-level quantum system in \S \ref{sec:7}
and show that the failure of Stapp's model is due to basic property
(concavity) of the von Neumann entropy of quantum systems, rather
than our failure to simulate the complexity of the real brain.

In conclusion, the present work shows that the breakdown of quantum
Zeno effect looks differently in Stapp's model (`random telegraph')
compared with no collapse version of quantum mechanics (`smear of
probabilities'). This difference, however, is not enough to make Stapp's
model plausible solution to the problem of environmental decoherence.
Instead, the future development of quantum theories of mind needs
to consider seriously the increase of quantum entropy due to the inevitable
coupling between the brain and its physical environment.

\section{\label{sec:2}Stapp's model}

Stapp describes the interaction between the mind and the brain with
the use of three basic processes 1, 2 and 3, attributed to John von Neumann \cite{vonNeumann1955}. In modern quantum
mechanical terminology these processes can be referred to as (1) projective
measurement, (2) unitary evolution and (3) wavefunction collapse.
Interestingly, Stapp always discusses these processes in the order
2, 1, 3 as they appear in his model of mind-brain interaction.

\subsection{\label{sub:Process2}Process 2}

The brain is considered to be an $n$-level quantum system whose states
belong to the Hilbert space $\mathcal{H}$. Unless the brain interacts
with the mind or the surrounding environment, the brain density matrix
$\hat{\rho}$ evolves according to the Schr\"{o}dinger equation:

\begin{equation}
\imath\hbar\frac{\partial}{\partial t}\hat{\rho}=\left[\hat{H},\hat{\rho}\right]
\end{equation}
where the brackets denote a commutator. If the Hamiltonian $\hat{H}$
is time-independent, the solution of the Schr\"{o}dinger equation
is given by:

\begin{equation}
\hat{\rho}\left(t\right)=e^{-\imath\hat{H}t/\hbar}\hat{\rho}\left(0\right)e^{\imath\hat{H}t/\hbar}\label{eq:Schroedinger}
\end{equation}
According to Stapp, Process 2 generates ``a cloud of possible worlds,
instead of the one world we actually experience'' \cite{Schwartz2005},
or ``a smear of classically alternative possibilities'' \cite{Stapp2007}.
Furthermore, Stapp claims that ``the automatic mechanical Process
2 evolution generates this smearing, and is in principle unable to
resolve or remove it'' \cite{Stapp2007}.

\subsection{Process 1}

The mind is able to perform repeated projective measurements upon
the brain using a freely chosen set of projection operators $\{\hat{P}_{1},\hat{P}_{2},\ldots,\hat{P}_{n}\}$,
which are mutually orthogonal $\hat{P}_{i}\hat{P}_{j}=\delta_{ij}\hat{P}_{j}$
and complete to identity $\sum_{j}\hat{P}_{j}=\hat{I}$. After each
projective measurement the brain density matrix undergoes non-unitary
transition:

\begin{equation}
\hat{\rho}\left(t\right)\rightarrow\sum_{j}\hat{P}_{j}\hat{\rho}\left(t\right)\hat{P}_{j}\label{eq:Stapp-transition}
\end{equation}
According to Stapp the ``Process 1 action extracts from {[}the{]}
jumbled mass of possibilities a particular {[}set{]} of alternative
possibilities'' among which only one is going to be actualized by
the Nature \cite{Stapp2005}.

\subsection{Process 3}

The actualization of only one possibility, from the set of available
possibilities, is done by the Nature. Colloquially, this process is
referred to as ``reduction of the wave packet''
or ``collapse of the wave function''.
Within the more general density matrix formalism, Process 3 is described
by a non-unitary transition that converts the unconditional density
matrix into conditional one (L\"{u}ders rule) \cite{Luders1951,Svensson2013}:

\begin{equation}
\sum_{j}\hat{P}_{j}\hat{\rho}\left(t\right)\hat{P}_{j}\rightarrow\frac{\hat{P}_{k}\hat{\rho}\left(t\right)\hat{P}_{k}}{\textrm{Tr}\left[\hat{P}_{k}\hat{\rho}\left(t\right)\right]}\label{eq:collapse}
\end{equation}
where $\hat{P}_{k}$ is a particular projector from the set $\{\hat{P}_{1},\hat{P}_{2},\ldots,\hat{P}_{n}\}$
selected by the Nature and $\textrm{Tr}\left[\hat{P}_{k}\hat{\rho}\left(t\right)\right]$
is the probability for the state to collapse to that particular state.

Here we remark that throughout this work we use the terms \emph{projective
measurements},\emph{ local projective measurements }and\emph{ local
projections} in the precise mathematical sense of \emph{Process 1}
given by Eq.~\ref{eq:Stapp-transition}. In no collapse models of
quantum mechanics the entanglement between the measured system and
the measuring apparatus results in the transition given by Eq.~\ref{eq:Stapp-transition}
and this is all there is in a physical measurement. In collapse models
of quantum mechanics the measurement is completed only after the entangled
state between the measured system and the measuring apparatus reduces
to a single outcome conditionally purifying the state and disentangling
the measured system from the measuring apparatus according to Eq.~\ref{eq:collapse}.
With our particular choice of terminology we allow the usage of the
term \emph{measurement} when discussing both collapse and no collapse
models of quantum mechanics and avoid trivial identification of \emph{measurement}
with \emph{wavefunction collapse} (for a general introduction to quantum
measurement see Ref.~\refcite{Svensson2013}).

\section{\label{sec:3}Quantum tunneling in a multiple-well as a quantum brain
model}

Biological implementation of Stapp's model should ultimately take
into account the basic electrophysiological processes that occur within
neurons in the human \emph{cerebral cortex} \cite{Kandel2012}. Each
neuron is composed of three different compartments: dendrites, soma
and axon (Fig.~\ref{fig:1}). Typically, the dendrites receive electric
inputs, the soma integrates the dendritic inputs, and the axon outputs
electric spikes that subsequently affect the electric properties of
dendrites of target neurons. At the molecular level, the electric
processes in neurons are regulated by opening and closing of voltage-gated
ion channels that are inserted in the plasma membrane.

\begin{figure}
\centerline{\psfig{file=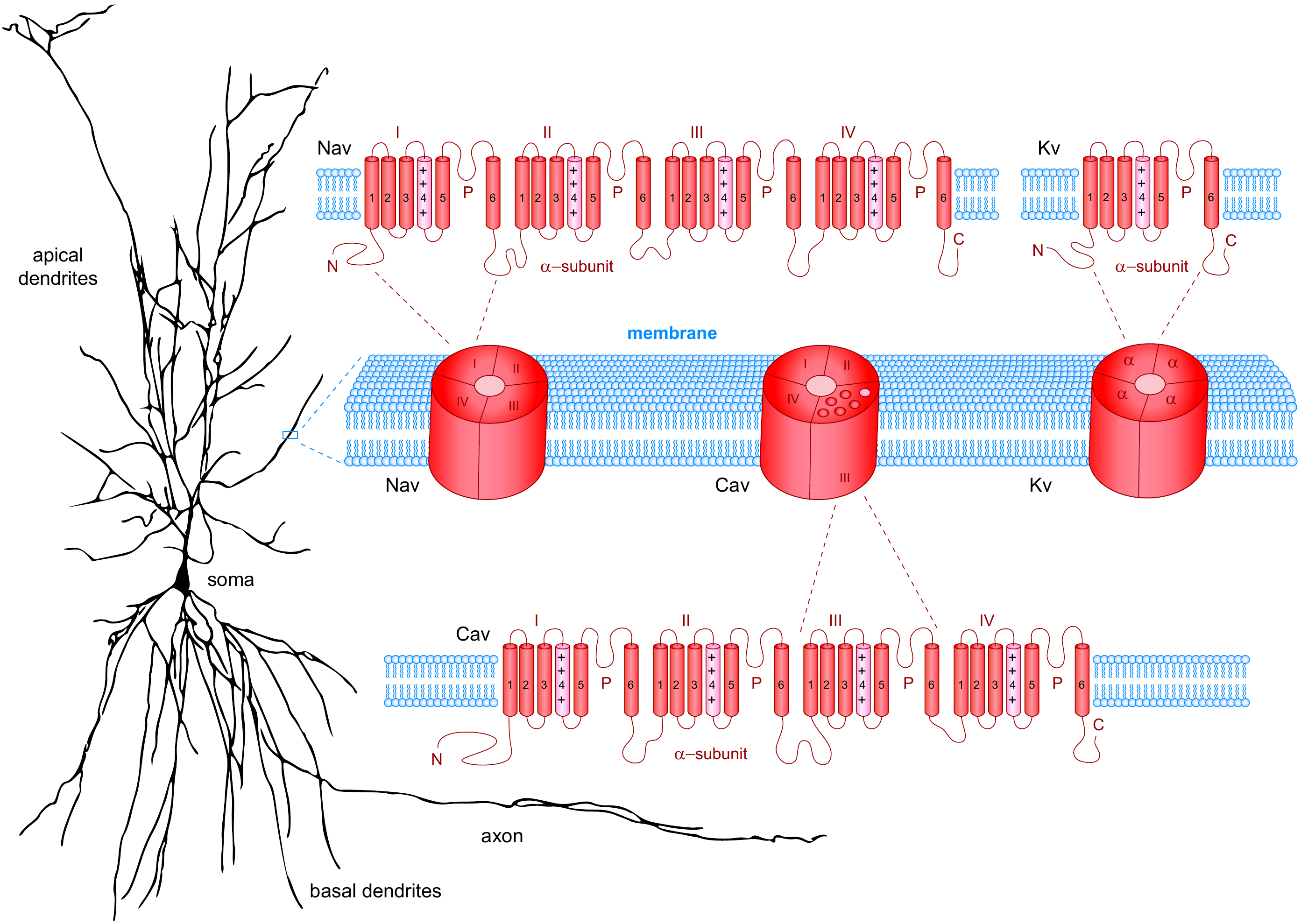,width=118mm}}
\vspace*{8pt}
\caption{\label{fig:1}Neuron morphology and common structure of voltage-gated
ion channels. The dendrites receive electric inputs that summate spatially
and temporally at the soma. If the transmembrane voltage at the axon
initial segment reaches a certain threshold of depolarization the
neuron fires an electric spike that propagates along the axon to affect
the dendrites of target neurons. Neuronal electric properties are
due to opening and closing of sodium (Nav), potassium (Kv) and calcium
(Cav) voltage-gated ion channels. Structurally, each channel is built
of four protein domains I-IV, each of which contains six transmembrane
$\alpha$-helices (1-6). The channel pore is formed by protein loops
(P) located between the 5th and 6th $\alpha$-helices, whereas the
voltage-sensing is performed by the 4th electrically charged $\alpha$-helix
within each domain.}
\end{figure}

Three families of voltage-gated ion channels are most abundant and
physiologically important in neurons: sodium (Nav), potassium (Kv)
and calcium (Cav) ion channels. All the three families of ion channels
share a common evolutionary conserved structure (Fig.~\ref{fig:1}).
Each channel is formed by a pore forming $\alpha$-subunit. The $\alpha$-subunit
of sodium (Nav) and calcium (Cav) channels is composed of four protein
domains I-IV, each of which contains 6 transmembrane $\alpha$-helices.
There is a minor difference in the structure of the potassium (Kv)
channels in which the protein domains I-IV are disconnected from each
other giving rise to four $\alpha$-subunits instead of a single one.
Structurally, the channel pore is formed by four protein loops (P)
located between the 5th and the 6th $\alpha$-helices of the protein
domains I-IV. The voltage-sensing is performed by a charged 4th $\alpha$-helix
of each domain \cite{Bezanilla2000,Bezanilla2005,Elliott2004,Gribkoff2009,Yarov2012,Jensen2012,Li2014}.

Macroscopic electric currents in neurons produced by voltage-gated
ion channels flow continuously across the plasma membrane depending
on the transmembrane voltage and the maximal channel conductance density.
The electric conductance through each individual channel, however,
can take only 2 discrete values: in the open conformation the channel
has a certain characteristic single channel conductance, whereas in
the closed conformation the conductance is zero. Single-channel recordings
have shown that at a given transmembrane voltage each voltage-gated
ion channel undergoes stochastic transitions between open and closed
states characterized by a certain probability for the given channel
type to be in the open conformation \cite{Sakmann1995}. For transmembrane
voltages that are far away from the threshold for generation of electric
spike, the behavior of brain neurons is insensitive to the small stochastic
fluctuations in the potential due to single-channel conformational
transitions. For transmembrane voltages near the threshold of $-55$
mV, however, opening or closing of a single channel can affect the
generation of an electric spike. In the cerebral cortex of humans
there are $\approx1.6\times10^{10}$ neurons \cite{Azevedo2009} and
these neurons can fire spikes with frequencies of $\approx40$ Hz.
Therefore, it is expected that each second in the human cerebral cortex
there are thousands of neurons that are sensitive to the opening or
closing of a single channel. Firing or not firing of these neurons
that are near the voltage threshold may have huge impact on cognitive
processes due to the highly nonlinear character of the cortical neuronal
networks.

Because the opening and closing of the voltage-gated ion channels is a stochastic process controlled by electron motion in the charged 4th $\alpha$-helix of each channel domain, it is biologically feasible to assume that Stapp's quantum Zeno model could be implemented via frequent measurements on the position of an electron in the voltage-sensing 4th $\alpha$-helix of ion channels. Here, we would like to underline the facts that the quantum state of the brain is generally unobservable, and that there is an upper bound on the classical information that can be extracted in a quantum measurement as shown by Alexander Holevo \cite{Holevo1973}. That is why the measured brain observable in the model, namely the position of the electron inside the voltage sensor of a neuronal ion channel, has been chosen to be physiologically paramount for the processing of information by real brains. In addition, the quantum tunneling of an electron in a multiple-well structure has already been considered to be an important toy model for studying quantum Zeno effect \cite{Gagen1993,Altenmuller1994}. For the subsequent numerical simulations we decided to use a symmetric triple-well potential \cite{Cole2008}, however it should be noted that the general algorithm for the simulations could be easily applied for more complex $n$-level systems and/or asymmetric potentials (see the Appendix).

Let us suppose that the normalized position states of the electron
in each of the wells are $|A\rangle$, $|B\rangle$ and $|C\rangle$,
there is no offset energy between different wells, and the tunneling
matrix elements between the states $|A\rangle$ and $|B\rangle$ or
$|B\rangle$ and $|C\rangle$ are equal $\kappa_{12}=\kappa_{23}$.
The Hamiltonian of the system in position basis is \cite{Cole2008}:

\begin{equation}
\hat{H}=\left(\begin{array}{ccc}
0 & -\kappa_{12} & 0\\
-\kappa_{12} & 0 & -\kappa_{23}\\
0 & -\kappa_{23} & 0
\end{array}\right)
\end{equation}
The energy eigenstates of the system are the eigenstates of the Hamiltonian.
If, under suitable choice of units, we set $\kappa_{12}=\kappa_{23}=\frac{1}{\sqrt{2}}$,
the energy eigenstates can be written as:

\begin{eqnarray}
|E_{\pm1}\rangle & = & \frac{1}{2}|A\rangle\mp\frac{1}{\sqrt{2}}|B\rangle+\frac{1}{2}|C\rangle\\
|E_{0}\rangle & = & \frac{1}{\sqrt{2}}|A\rangle-\frac{1}{\sqrt{2}}|C\rangle
\end{eqnarray}
with eigenvalues $E_{\pm1}=\pm1$ and $E_{0}=0$.

With the use of the projectors $\hat{P}_{J}=|J\rangle\langle J|$,
$J\in\{A,B,C\}$, we can express the probability to find the electron
in well $J$ at time $t$ as:

\begin{equation}
p_{J}\left(t\right)=\textrm{Tr}\left[\hat{P}_{J}e^{-\imath\hat{H}t/\hbar}\hat{\rho}\left(0\right)e^{\imath\hat{H}t/\hbar}\right]\label{eq:prob}
\end{equation}
If at $t=0$ the electron is located in well $A$, the initial density
matrix of the system is:

\begin{equation}
\hat{\rho}\left(0\right)=\left(\begin{array}{ccc}
1 & 0 & 0\\
0 & 0 & 0\\
0 & 0 & 0
\end{array}\right)
\end{equation}
and the probabilities for detection in each of the three wells at
time $t$ calculated from Eq.~\ref{eq:prob} are:

\begin{gather}
p_{A}\left(t\right)=\frac{1}{4}\left[\cos\left(t/\hbar\right)+1\right]^{2}\\
p_{B}\left(t\right)=\frac{1}{2}\left[\sin\left(t/\hbar\right)\right]^{2}\\
p_{C}\left(t\right)=\frac{1}{4}\left[\cos\left(t/\hbar\right)-1\right]^{2}
\end{gather}
The electron tunnels coherently forth-and-back from well $A$ to wells
$B$ and $C$ as illustrated in Fig.~\ref{fig:2}. Such behavior
is consistent with previous works \cite{Cole2008,Gagen1993} and shows
that the quantum system in a multiple potential well, if unperturbed,
is able to return periodically to its initial state.

\begin{figure}[t]
\centerline{\psfig{file=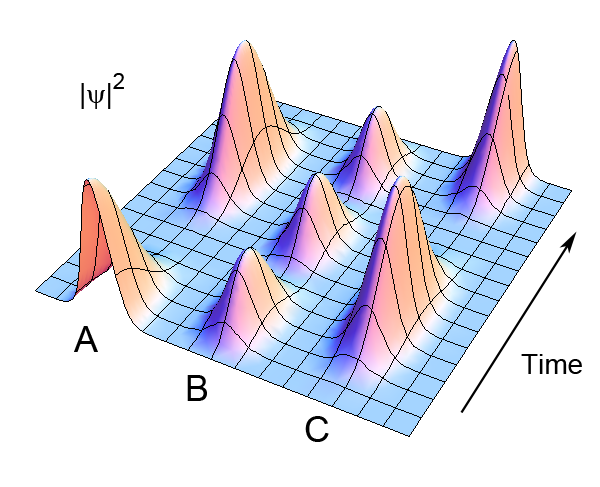,width=76mm}}
\vspace*{8pt}
\caption{\label{fig:2}Quantum tunneling of the brain state in a triple well
potential simulated for a period of time $t=3\pi/\hbar$. At times
$t=2k\pi/\hbar$, $k=0,1,2,\ldots$ the brain state is localized in
well $A$, whereas at times $t=(2k+1)\pi/\hbar$, the brain state
is localized in well $C$. If unperturbed, the quantum system in a
multiple potential well tunnels coherently forth-and-back from well
$A$ to wells $B$ and $C$ and returns periodically to its initial
state. The probability $|\psi|^{2}$ is normalized so that $\int|\psi|^{2}dx=1$.}
\end{figure}

\section{\label{sec:4}Monte Carlo simulation of quantum Zeno effect}

If the electron in the voltage-sensing 4th $\alpha$-helix of an ion
channel is initially in an eigenstate of the position operator, it
would be possible to achieve quantum Zeno effect using projective
measurements in a position basis. Let us suppose that the electron
is initially in well $A$ and we perform repeated projective measurements
with the projectors $\hat{P}_{J}$ at regular time intervals $\xi$.
Between the projective measurements the electron state evolves coherently
according to the Schr\"{o}dinger equation (Process 2). After each
measurement the density matrix $\hat{\rho}$ of the electron in the
quantum brain is diagonalized in the position basis (Process 1). To
achieve ``reduction of the wave packet'' (Process 3), we implement
\emph{weighted Random Choice} in \emph{Wolfram's Mathematica} \emph{9},
where one of the pure state density matrices $\hat{\rho}_{A}=|A\rangle\langle A|$,
$\hat{\rho}_{B}=|B\rangle\langle B|$ or $\hat{\rho}_{C}=|C\rangle\langle C|$
is randomly chosen with corresponding weights $p_{A}\left(\xi\right)$,
$p_{B}\left(\xi\right)$ and $p_{C}\left(\xi\right)$ and the state
of the quantum brain is updated accordingly. In the limit $\xi\to0$,
the electron in the quantum brain stays with probability of $1$ in
its initial state. A single trial from the Monte Carlo simulation
with a non-zero $\xi=\pi/8\hbar$ is shown in Fig.~3a. The quantum
state of the brain remains in the well $A$ for a period of time equal
to $19$ $\xi$-steps before jumping to well $B$, where it stays
until the end of the simulation. We note that while running multiple
Monte Carlo trials the majority of them achieve the intended quantum
Zeno effect, still there are trials in which the state jumps early
to well $B$ and stays there instead of well $A$. The results show
that in the absence of environmental decoherence the quantum Zeno
effect indeed could be achieved with a certain efficiency that is
proportional to the average time for which the decay of the initial
state is suppressed. The time point at which the quantum Zeno effect
breaks down is best visualized on a plot showing the statistical average
of multiple Monte Carlo trials (Fig.~3b). The breakdown of the quantum
Zeno effect could be understood as the drop of the \emph{unconditional
probability} for the system to be in the initial state (in this case
well $A$) to a value $\leq\frac{1}{2}$. Since the efficiency of
any scheme attempting to achieve quantum Zeno effect is estimated
by calculation of unconditional density matrices and unconditional
probabilities, it is clear that the presence of wavefunction collapses
(reductions) is irrelevant and cannot help repairing a faulty quantum
Zeno scheme.

\begin{figure}[t]
\centerline{\psfig{file=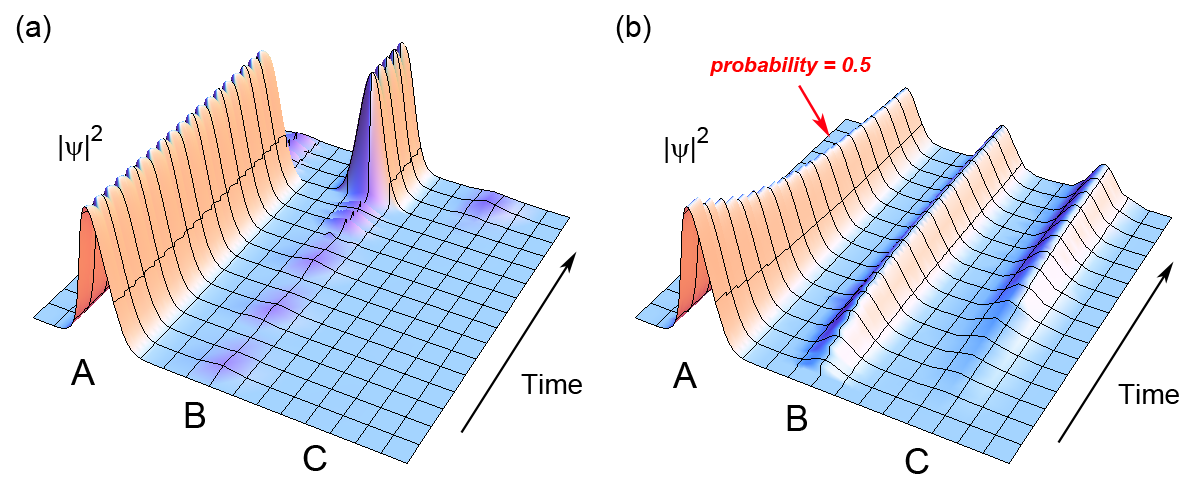,width=118mm}}
\vspace*{8pt}
\caption{Quantum Zeno effect in a triple well potential simulated for a period
of time $t=3\pi/\hbar$. Mind efforts perform projective position
measurements of the brain state at times separated by time interval
$\xi=\pi/8\hbar$ resulting in suppressed evolution of the brain state.
(a) The result from a single Monte Carlo simulation trial successfully
achieving quantum Zeno effect manifested as suppressed decay of the
initial state $|A\rangle$. The quantum state of the brain starts
in well $A$ and remains there for a period of time equal to $19$
$\xi$-steps before jumping to well $B$. (b) Statistical average
of multiple Monte Carlo trials reproduces the probabilities contained
in the unconditional density matrix of the system. This result can
also be interpreted as the state of the multiverse in no-collapse
models of quantum mechanics. Red arrow indicates the time point at
which the quantum Zeno effect breaks down--the probability for the
electron \emph{not to be} in well $A$ is at least equal to the probability
\emph{to be} in well $A$. The probability $|\psi|^{2}$ is normalized
so that $\int|\psi|^{2}dx=1$.}
\end{figure}

\section{\label{sec:5}Decoherence induced breakdown of quantum Zeno effect}

Now, we are ready to address the main claim made by Stapp, whether
mind efforts could achieve quantum Zeno effect upon the brain even
in the presence of environmentally induced decoherence of the brain.
In general, the effect of environmental decoherence is to diagonalize
the density matrix $\hat{\rho}$ in a certain basis (so called pointer
basis) depending on the interaction Hamiltonian $\hat{H}_{\textrm{int}}$
that describes the coupling between the brain and its environment.
Because Stapp claims that his model is robust against the effects
of environmental decoherence, we could choose $\hat{H}_{\textrm{int}}$
such that it diagonalizes the density matrix of the brain $\hat{\rho}$
in a basis different from the position basis. Here, we take the pointer
basis to be the energy basis -- that is we assume that the brain undergoes
dephasing in the process of its interaction with the environment \cite{Venugopalan1998,Hornberger2009}.
Alternatively, identical results will be obtained if we start from
a brain that is in an eigenstate of the energy basis, require the
mind to exert quantum Zeno effect in the energy basis, and consider
environmental decoherence in position basis. If $\tau$ is the decoherence
time of the brain \cite{Tegmark2000}, the action of the environmental
decoherence could be modeled using non-unitary transition:

\begin{equation}
\hat{\rho}\left(\tau\right)\to\sum_{j}\hat{P}_{E_{j}}\hat{\rho}\left(\tau\right)\hat{P}_{E_{j}}
\end{equation}
where the projectors $\hat{P}_{E_{j}}=|E_{j}\rangle\langle E_{j}|$,
$j\in\{\pm1,0\}$. Next, to achieve maximal efficiency of mind efforts,
we let the mind perform repeated projective measurements with the
projectors $\hat{P}_{J}$ at regular time intervals $\xi\rightarrow0$.
The weights in the Random Choice function at times separated by intervals
$\tau$ are as follows: if $\hat{\rho}\left(\tau\right)=\hat{\rho}_{A}$
or $\hat{\rho}\left(\tau\right)=\hat{\rho}_{C}$ then $p_{A}\left(\tau\right)=\frac{3}{8}$,
$p_{B}\left(\tau\right)=\frac{1}{4}$, $p_{C}\left(\tau\right)=\frac{3}{8}$
and if $\hat{\rho}\left(\tau\right)=\hat{\rho}_{B}$ then $p_{A}\left(\tau\right)=\frac{1}{4}$,
$p_{B}\left(\tau\right)=\frac{1}{2}$, $p_{C}\left(\tau\right)=\frac{1}{4}$.
These weights could be calculated from an initial position eigenstate
that undergoes a projective measurement in the energy basis using
the projectors $\hat{P}_{E_{j}}$, followed by a projective measurement
into the position basis using the projectors $\hat{P}_{J}$. A typical
result from the Monte Carlo simulation of mind efforts attempting
to achieve quantum Zeno effect in the presence of environmental decoherence
is shown in Fig.~4a. For timescales greater than the decoherence time
$\tau$, the quantum state of the brain randomly jumps between the
different wells -- a behavior previously described as `random telegraph'
\cite{Gagen1993}. The statistical average of multiple Monte Carlo
trials shown in Fig.~4b does not show the individual random walks produced
by the collapses. Nevertheless the probability distribution in Fig.~4b
unambiguously demonstrates that the breakdown of the quantum Zeno
effect occurs at a timescale comparable with the decoherence time
$\tau$. For $t\geq\tau$ the probability for the quantum brain to
be in any of the potential wells becomes equal to $\frac{1}{3}$ --
a state described by unconditional density matrix with maximal von
Neumann entropy. Realization of a quantum Zeno effect-like walk in
which the state stays in well $A$ for $t=24\tau$ will occur with
probability of $\left(\frac{1}{3}\right)^{23}=1.06\times10^{-11}$.

\begin{figure}[t]
\centerline{\psfig{file=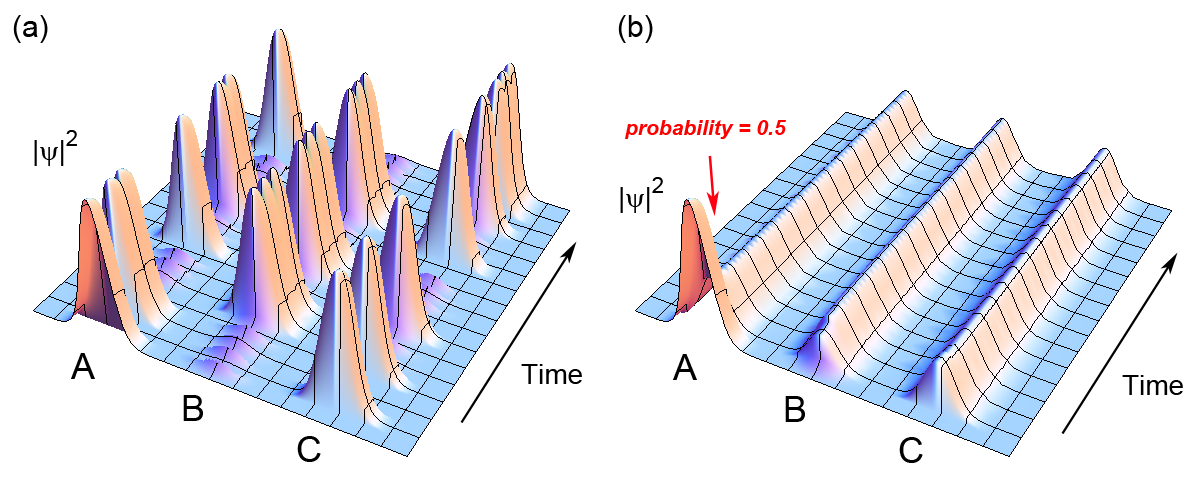,width=118mm}}
\vspace*{8pt}
\caption{Decoherence induced breakdown of quantum Zeno effect simulated for
a period of time $t=24\tau$ during which mind efforts perform projective
position measurements of the brain at times separated by time interval
$\xi\rightarrow0$. (a) The result from a single typical Monte Carlo
simulation trial. Decoherence induces random jumps of the quantum
brain state within the triple well potential--the brain behaves as
a `random telegraph'. The quantum Zeno effect does not persist for
timescales greater than the decoherence time $\tau$. (b) Statistical
average of multiple Monte Carlo trials reproduces the probabilities
contained in the unconditional density matrix of the system. This
result can also be interpreted as the state of the multiverse in no-collapse
models of quantum mechanics. Red arrow indicates the time point at
which the quantum Zeno effect breaks down--the probability for the
electron \emph{not to be} in well $A$ is at least equal to the probability
\emph{to be} in well $A$. The probability $|\psi|^{2}$ is normalized
so that $\int|\psi|^{2}dx=1$.}
\end{figure}

The above results affirm that to study the quantum Zeno effect in
the brain, one does not really need to model the collapse (Process
3) as claimed by Stapp \cite{Stapp2001,Stapp2005,Schwartz2005,Stapp2007,Stapp2012},
because in order to produce the unconditional density matrix one has
to average over the results from individual collapses. If the mind
action is not followed by a collapse, and provided that the mind probing
actions and environmental decoherence do not occur in the same basis
(that is the two bases do not have shared basis vectors), for time
$t\geq\tau$ the unconditional density matrix of the brain will tend
to one with maximal von Neumann entropy (in which the probability
to find the brain in any state is $\frac{1}{n}$) . The result from
a simulation in which the collapse (Process 3) is omitted is identical
with the plot shown in Fig.~4b, and thus provides all the necessary
information needed for one to find the time point at which the \emph{unconditional
probability} for the system to be in the initial state drops to a
value $\leq\frac{1}{2}$.

\section{\label{sec:6}Local projections and brain density matrix}

The Monte Carlo simulations of the $n=3$ level quantum system could
be easily performed on a personal computer. Repeating the algorithm
for large $n$, however, would need the processing power of a supercomputer
and substantial financial and time investment. It is thus desirable
to prove as a theorem that the Monte Carlo results obtained for $n=3$
would hold as well for any $n$. To achieve this we will estimate
the change of the von Neumann entropy of the unconditional brain
density matrix under local projections.
\begin{definition}
The\emph{ }von Neumann entropy\emph{ }of a quantum-mechanical system
described by a density matrix $\hat{\rho}$ is:

\begin{equation}
S(\hat{\rho})=-\textrm{Tr}\left(\hat{\rho}\ln\hat{\rho}\right)=-\sum_{i}\lambda_{i}\ln\lambda_{i}
\end{equation}
\end{definition}
\begin{theorem}
\label{Lemma:1}The von Neumann entropy is invariant under unitary
evolution \cite{Wehrl1978,Bengtsson2006,Ohya2010}:
\begin{equation}
S(\hat{U}\hat{\rho}\hat{U}^{\dagger})=S(\hat{\rho})
\end{equation}
\end{theorem}

\begin{theorem}
\label{Lemma:2}The von Neumann entropy is a concave functional \cite{Wehrl1978,Bengtsson2006,Ohya2010,Delbruck1936,Lieb1975}.
If $p_{1},p_{2},\ldots,p_{n}\geq0$ \textup{and $\sum_{i}p_{i}=1$,
then:}

\begin{equation}
S(\sum_{i}p_{i}\hat{\rho}_{i})\geq\sum_{i}p_{i}S(\hat{\rho}_{i})
\end{equation}
\end{theorem}

We have argued previously that if the mind efforts were to act only
locally at the brain using projection operators, then such action
cannot decrease the von Neumann entropy of the brain density matrix
\cite{Georgiev2012}. Because decoherence increases the von Neumann
entropy over time, it would follow as a corollary that quantum Zeno
effect cannot be achieved via local projections in the presence of
environmental decoherence. Our previous argument, however, utilized
a two-level approximation and a pure initial state. Stapp objected
that the argument is based on an improper extrapolation of a theorem
valid for a two-level system to the much more complicated $n$-level
system of the real brain \cite{Stapp2012}. Here, we prove a generalization
of our previous theorem, which is valid for any $n$-level quantum
system and any purity of the initial density matrix $\hat{\rho}$
of that system.
\begin{theorem}
\label{Theo:1}Local projective measurements upon a quantum system
$\mathbb{Q}$ using a freely chosen set of projection operators $\{\hat{P}_{1},\hat{P}_{2},\ldots,\hat{P}_{n}\}$,
which are mutually orthogonal $\hat{P}_{i}\hat{P}_{j}=\delta_{ij}\hat{P}_{j}$
and complete to identity $\sum_{j}\hat{P}_{j}=\hat{I}$ cannot decrease
the von Neumann entropy of the unconditional density matrix $\hat{\rho}$
of the system $\mathbb{Q}$. Moreover, the entropy cannot be decreased
for any subspace of the density matrix $\hat{\rho}$.\end{theorem}
\begin{proof}
Represent the initial density matrix $\hat{\rho}_{0}$ in the basis
in which the projection operators $\{\hat{P}_{1},\hat{P}_{2},\ldots,\hat{P}_{n}\}$
are expressed in their simplest form (a single unit on the diagonal
with zeros elsewhere). If $\hat{\rho}_{0}$ is diagonal in that basis,
it will remain unchanged by the action of the projectors and the entropy
will stay the same. In general, however, $\hat{\rho}_{0}$ will not
be diagonal:

\begin{alignat}{1}
\hat{\rho}_{0} & =\left(\begin{array}{ccccc}
a_{11} & a_{12} & a_{13} & \cdots & a_{1n}\\
a_{12}^{*} & a_{22} & a_{23} & \cdots & a_{2n}\\
a_{13}^{*} & a_{23}^{*} & a_{33} & \cdots & a_{3n}\\
\vdots & \vdots & \vdots & \ddots & \vdots\\
a_{1n}^{*} & a_{2n}^{*} & a_{3n}^{*} & \cdots & a_{nn}
\end{array}\right)\label{eq:initial-matrix}
\end{alignat}
The action of the projectors (see Stapp's Process 1 given by Eq.~\ref{eq:Stapp-transition})
will be to kill all off-diagonal entries:

\begin{equation}
\sum_{j}\hat{P}_{j}\hat{\rho}\left(t\right)\hat{P}_{j}=\hat{\rho}_{n}=\left(\begin{array}{ccccc}
a_{11} & 0 & 0 & \cdots & 0\\
0 & a_{22} & 0 & \cdots & 0\\
0 & 0 & a_{33} & \cdots & 0\\
\vdots & \vdots & \vdots & \ddots & \vdots\\
0 & 0 & 0 & \cdots & a_{nn}
\end{array}\right)
\end{equation}
To show that the entropy of $\hat{\rho}_{n}$ is larger than the entropy
of $\hat{\rho}_{0}$, we construct a chain of inequalities. First,
note that we can evolve unitarily $\hat{\rho}_{0}$ in two different
ways and sum the results to kill all off-diagonal entries in the first
row and column:

\begin{eqnarray}
\hat{\rho}_{1} & = & \left(\begin{array}{ccccc}
a_{11} & 0 & 0 & \cdots & 0\\
0 & a_{22} & a_{23} & \cdots & a_{2n}\\
0 & a_{23}^{*} & a_{33} & \cdots & a_{3n}\\
\vdots & \vdots & \vdots & \ddots & \vdots\\
0 & a_{2n}^{*} & a_{3n}^{*} & \cdots & a_{nn}
\end{array}\right)\nonumber \\
 & = & \frac{1}{2}\left(\begin{array}{ccccc}
-1 & 0 & 0 & \cdots & 0\\
0 & 1 & 0 & \cdots & 0\\
0 & 0 & 1 & \cdots & 0\\
\vdots & \vdots & \vdots & \ddots & \vdots\\
0 & 0 & 0 & \cdots & 1
\end{array}\right)\hat{\rho}_{0}\left(\begin{array}{ccccc}
-1 & 0 & 0 & \cdots & 0\\
0 & 1 & 0 & \cdots & 0\\
0 & 0 & 1 & \cdots & 0\\
\vdots & \vdots & \vdots & \ddots & \vdots\\
0 & 0 & 0 & \cdots & 1
\end{array}\right)+\frac{1}{2}\hat{I}\hat{\rho}_{0}\hat{I}
\end{eqnarray}
Using Theorems \ref{Lemma:1} and \ref{Lemma:2}, we obtain
$S(\hat{\rho}_{1})\geq S(\hat{\rho}_{0})$. Next, we can evolve unitarily
$\hat{\rho}_{1}$ in two different ways and sum the results to kill
its off-diagonal entries in the second row and column:

\begin{eqnarray}
\hat{\rho}_{2} & = & \left(\begin{array}{ccccc}
a_{11} & 0 & 0 & \cdots & 0\\
0 & a_{22} & 0 & \cdots & 0\\
0 & 0 & a_{33} & \cdots & a_{3n}\\
\vdots & \vdots & \vdots & \ddots & \vdots\\
0 & 0 & a_{3n}^{*} & \cdots & a_{nn}
\end{array}\right)\nonumber \\
 & = & \frac{1}{2}\left(\begin{array}{ccccc}
1 & 0 & 0 & \cdots & 0\\
0 & -1 & 0 & \cdots & 0\\
0 & 0 & 1 & \cdots & 0\\
\vdots & \vdots & \vdots & \ddots & \vdots\\
0 & 0 & 0 & \cdots & 1
\end{array}\right)\hat{\rho}_{1}\left(\begin{array}{ccccc}
1 & 0 & 0 & \cdots & 0\\
0 & -1 & 0 & \cdots & 0\\
0 & 0 & 1 & \cdots & 0\\
\vdots & \vdots & \vdots & \ddots & \vdots\\
0 & 0 & 0 & \cdots & 1
\end{array}\right)+\frac{1}{2}\hat{I}\hat{\rho}_{1}\hat{I}
\end{eqnarray}
Again, using Theorems \ref{Lemma:1} and \ref{Lemma:2},
we obtain $S(\hat{\rho}_{2})\geq S(\hat{\rho}_{1})$. Repeating the
construction $n$ times gives us the bound 
\begin{equation}
S(\hat{\rho}_{n})\geq S(\hat{\rho}_{n-1})\geq\ldots\geq S(\hat{\rho}_{2})\geq S(\hat{\rho}_{1})\geq S(\hat{\rho}_{0})
\end{equation}
The above construction can be applied as well for every subspace of
the density matrix from which follows that there is no subspace in
which the change of the entropy is negative after projective measurement.
\end{proof}

\section{\label{sec:7}Implications for Stapp's model}

Wavefunction collapses (Process 3) produce pure states that are described
by conditional density matrices with zero quantum entropy. Conditional
density matrices however cannot assess whether a given quantum Zeno
scheme is plausible or not -- for that one needs unconditional probabilities.
A classical example nicely illustrates this point: suppose that one
buys a lottery ticket that has a 1 in a million chance to win the
jackpot. It is wrong to claim that buying the lottery ticket is an
efficient or plausible way to become a millionaire based on a conditional
reasoning such as `if you happen to win the lottery then there is
an absolute certainty (probability of 1) that you will become a millionaire'.
Instead, the efficiency of buying the lottery ticket is 1 in a million
based on the unconditional probability for the event. Thus, the main
result provided by Theorem \ref{Theo:1} concerning the entropy
of the unconditional density matrix of the brain is exactly tailored
to assess the plausibility of Stapp's quantum Zeno scheme.

The implications of Theorem \ref{Theo:1} for Stapp's model
become clearer if we recognize that the enviromental decoherence is
just a form of quantum Zeno effect in a so-called pointer basis \cite{Joos1984,Joos1985,Joos2007}.
If the decoherence pointer basis happens to be the one in which the
mind action is intended, then the decoherence alone can achieve the
quantum Zeno effect and the mind action will be irrelevant or redundant.
Thus, in order for mind's action to be functionally meaningful it
needs to achieve something different that is not already achieved
by the brain environment.

First, suppose that the brain does not possess \emph{decoherence-free
subspace}. Trivially, quantum Zeno effect can be achieved for states
located within decoherence-free subspace because the environmental
action for these states is absent. However, Stapp explicitly states
that quantum Zeno effect could work even if the brain coherence is
destroyed due to environmental interaction, hence regardless the lack
of decoherence-free subspace. This is also evident in Stapp's own
illustrations \cite[Figs.~11.2-11.6]{Stapp2007} showing that decoherence
reduces the density matrix of the brain to nearly diagonal form in
the position basis. Since the density matrix $\hat{\rho}$ is a Hermitian
operator, it always has a representation basis in which it is diagonal.
If only environmental interaction is considered, for times larger
than the decoherence time $\tau$ one may identify the basis in which
$\hat{\rho}$ is diagonal with the pointer basis of environmental
decoherence.

Second, suppose that the mind intends to keep the probability $p(a_{1},t_{0})>\frac{1}{2}$
of a given brain state $|a_{1}\rangle\langle a_{1}|$ as high as possible
at a later time $t>t_{0}$ in the presence of environmental decoherence.
Breakdown of quantum Zeno effect will be signaled by $p(a_{1},t_{0})]\leq\frac{1}{2}$
because at this point the probability for the brain \emph{not to be}
in the given state is at least equally, if not more, probable than
being in the given state.

We can now show that if the vector $|a_{1}\rangle$ does not belong
to the pointer basis, the repeated local projective measurements
of the brain by the mind using the projector $|a_{1}\rangle\langle a_{1}|$ will
always have probability $p(a_{1},t)$ that is lower or equal to the
probability $p^{\prime}(a_{1},t)$ for the case in which the mind
does not perform any measurement.

Consider an initial brain density matrix that is diagonalized in the
pointer basis. Let the mind perform two projective measurements of
the brain at times $t_{1}$ and $t_{2}$ using the same projector
$|a_{1}\rangle\langle a_{1}|$ (as part of a complete set of projectors)
in an attempt to keep the brain in this state and preserve the probability
$p(a_{1},t)>\frac{1}{2}$ for as long period of time as possible.
From Eq.~\ref{eq:Stapp-transition} it follows that at both times
$t_{1}$ and $t_{2}$ the density matrix $\hat{\rho}$ will be diagonal
in the basis chosen by the mind. From Theorem \ref{Theo:1}
follows that $S\left[\hat{\rho}(t_{2})\right]\geq S\left[\hat{\rho}(t_{1})\right]$
and similar inequalities hold for all subspaces of $\hat{\rho}$. Here we remind that the von Neumann entropy is directly related to the eigenvalues of the density matrix and that the eigenvalues are exhibited on the main diagonal of the density matrix when the matrix is diagonal.
Thus, using the facts that the matrices are diagonal, the function $x\ln x$
is concave in the interval $[0,1]$, and $p(a_{1},t_{0})>\frac{1}{2}$,
we obtain $p(a_{1},t_{2})\leq p(a_{1},t_{1})$. Because repeated local
projections instead of increasing the probability of $|a_{1}\rangle\langle a_{1}|$
over the initial $p(a_{1},t_{0})$ can only speed up the probability
decay, it follows that the mind efforts can only lead to extra decay of
probability in addition to the decay due to environmental decoherence.
Indeed, if there are no shared basis vectors between the pointer basis and
the basis chosen by the mind for the projective measurements, for
$t\geq\tau$ the entropy of the brain density matrix will tend to
the maximal one:

\begin{equation}
S(\hat{\rho})_{\textrm{max}}=S\left[\left(\begin{array}{cccc}
\frac{1}{n} & 0 & \cdots & 0\\
0 & \frac{1}{n} & \cdots & 0\\
\vdots & \vdots & \ddots & \vdots\\
0 & 0 & \cdots & \frac{1}{n}
\end{array}\right)\right]=\ln n
\end{equation}
If the basis chosen by the mind happens to be mutually unbiased with
the decoherence pointer basis (that is the inner product of any two
vectors from the two bases is equal to $\frac{1}{n}$), then the decay
of the initial probability from $p(a_{1},t_{0})$ to $\frac{1}{n}$
will be fastest and will occur immediately after the first projective
measurement performed by the mind.

The above result is devastating for Stapp's model because instead
of achieving quantum Zeno effect, the mind efforts can only speed
up the decay of probability for the brain to be in a given state.
In other words, in the presence of environmental decoherence the best
strategy of the mind is to \emph{withhold} performing any projective
measurements. Because Theorem \ref{Theo:1} follows from the
standard Hilbert space formalism of quantum mechanics, any amendment
done to Stapp's model will necessarily result in a new physical theory
that is inconsistent with quantum mechanics.

\section{Discussion}

The linear unitary time evolution of quantum systems given by the Schr\"{o}dinger equation (Eq.~\ref{eq:Schroedinger}) preserves inner products. This implies, however, that quantum systems prepared in a superposition of states could interact with macroscopic objects such as cats forcing them into a superposition of mutually exclusive alternatives such as ``dead cat'' or ``alive cat'' \cite{Schrodinger1935}. To explain the apparent lack of macroscopic superpositions in the surrounding world, two different approaches have been proposed. The first approach is to counterintuitively accept the linear unitary time evolution as fundamental, and exclude the collapse postulate from the list of fundamental quantum mechanical axioms. In such no collapse models of quantum mechanics all possible outcomes do actually get realized in different universes consistently with the Born rule, and the sum of all decoherent universes forms a multiverse \cite{Everett2012,Zeh2013,Mensky2013}. The second approach is to accept the collapse postulate as a description of an objective physical process and introduce nonunitary time evolution of quantum systems \cite{Melkikh2014,Penrose2014,Ghirardi1986}. In this work, we have considered both types of approaches. First, the results obtained for the unconditional density matrix of the brain could be directly interpreted as describing the state of the multiverse in no collapse models of quantum mechanics. Second, Monte Carlo simulations of collapse models of quantum mechanics were performed using the L\"{u}ders rule (Eq.~\ref{eq:collapse}) that provides an effective mathematical description of the objective collapse regardless of the exact nature of the underlying physical process. Both collapse and no collapse approaches are restrained by the quantum informational theorem \ref{Theo:1} because, if the probabilities for all collapse outcomes are consistent with the Born rule, one could statistically average over all possible outcomes obtained from Eq.~\ref{eq:collapse} to get the very same unconditional density matrix that is predicted by no collapse models.

The mathematical description of the basic postulates entering into
Stapp's model appears to be similar to previous works on foundations
of quantum mechanics \cite{vonNeumann1955,Luders1951}. However, Stapp
surprisingly claimed that his model is robust against environmental
decoherence and that mind efforts could exert quantum Zeno effect
upon the brain for timescales longer than the decoherence time of
the brain $\tau$ \cite{Stapp2001,Stapp2005,Schwartz2005,Stapp2007,Stapp2007b}.
Stapp attributed the putative success of his model to objective collapses
(Process 3) resulting from the mind probing action. To directly test
this claim, we performed Monte Carlo simulations of quantum Zeno effect
in the brain in the presence or absence of environmental decoherence.
The results from the simulations unambiguously show that if the environmental
decoherence happens to be in a basis different from the one chosen
by the mind, the quantum Zeno effect is lost for timescales greater
than the decoherence time $\tau$. Furthermore, the conditional purification
of the brain density matrix due to objective collapses cannot remedy
the faulty quantum Zeno effect attempted by the mind.

Because every attempt to computationally simulate the brain could
be objected on the grounds that it is too simplistic to capture the
full complexity of the real brain, we have shown that the breakdown
of quantum Zeno effect is due to fundamental property (concavity)
of quantum entropy. Namely, repeated projective measurements cannot
decrease the von Neumann entropy $S(\hat{\rho})$ of the unconditional
density matrix of the measured system but may only increase it. The
increase of the entropy is particularly pronounced in the case when
the measured system is subject to environmental decoherence in a basis
different from the one of the performed projective measurements. If
the two bases are mutually unbiased the entropy reaches the maximum
of $\ln n$ within a period of time equal to the decoherence time
$\tau$. Only in the special case when both the mind and the environment
perform projective measurements upon the brain in the same basis,
one could expect to observe quantum Zeno effect for timescales larger
than the decoherence time $\tau$. Such scenario, however, makes the
mind efforts useless from a functional viewpoint, because the action
of the mind becomes redundant with the action of the environment.
In addition, there could be no free will if the mind cannot choose
the basis in which the projective measurement is performed. Since
the efficacy of mind efforts to exert quantum Zeno effect upon the
brain depend strongly on the basis in which the environmental decoherence
occurs, Stapp's model does not appear to be physically
plausible. 

Theorem \ref{Theo:1} sets an important constraint on the future development
of quantum theories of mind and shows that quantum Zeno effect cannot
work if the environmental decoherence affects the whole Hilbert space
of the brain. In the presence of decoherence-free subspace, however,
one could use quantum Zeno effect to combat some of the negative effects
of decoherence on probability decay using unitary operations that
first `hide' the initial state in the decoherence-free subspace and
at a later time restore the state of interest using inverse unitary
operations \cite{Lidar1998,Kim2012}. The latter admittedly speculative
possibility brings us back to the problem of decoherence albeit in
a mitigated form: one needs to explain not why environmental decoherence
does not affect the whole brain, but why environmental decoherence
does not affect a subspace of the brain. Since such an explanation
is likely dependent on the specific architecture of the brain, to
find it a tighter interdisciplinary collaboration and further research
crossing the boundaries of quantum physics and molecular neuroscience
would be needed.

\section*{Acknowledgements}

I would like to thank Professor James F. Glazebrook (Department of
Mathematics, Eastern Illinois University) for helpful comments and
bringing to my attention previously published results on the concavity
of von Neumann entropy.

\appendix{Algorithm for the Monte Carlo simulations}

The main steps in the Monte Carlo simulations of Stapp's model, that
could be used for any $n$-level system and any multiple-well potential,
are:
\begin{enumerate}
\item Express the Hamiltonian $\hat{H}$ of the system in a matrix form
using the position basis.
\item Calculate the (energy) eigenvectors $|E\rangle$ and eigenvalues $E$
of the Hamiltonian $\hat{H}$.
\item Write the projectors in the position basis $\hat{P}_{J}=|J\rangle\langle J|$
(trivial) and the projectors in energy basis $\hat{P}_{E}=|E\rangle\langle E|$.
\item Between projective measurements evolve the density matrix $\hat{\rho}$
using the Schr\"{o}dinger equation (Eq.~\ref{eq:Schroedinger}).
\item At times separated by time steps $\xi$ apply the projectors $\hat{P}_{E}\hat{P}_{J}$
according to Eq.~\ref{eq:Stapp-transition} and discontinuously
update the density matrix to one of the pure density matrices $|J\rangle\langle J|$
using weighted Random Choice function with corresponding weight $p_{J}\left(\xi\right)$.
\item Plot the results graphically.
\end{enumerate}
Even though it is straightforward to implement the algorithm to more
complex $n$-level systems, there will be several drawbacks. First,
the visual comprehension of the plotted results in larger number of
dimensions $n$ becomes difficult and the plots would lose their didactic
utility in explaining why objective wavefunction collapses cannot
help Stapp's quantum Zeno mind-brain model in the presence of environmental
decoherence. Second, for large $n$ one needs to numerically approximate
the energy eigenvectors and eigenvalues instead of providing concise
analytical expressions as in the case of symmetric $3$-level system.
The resulting lengthy numerical formulas will distract the reader
rather than highlight the conceptual issues at stake. Third, any computationally
accessible $n$ would inevitably be considered too small to account
for the complexity of the real brain. To avoid that problem, we have
proven a quantum informational theorem that allows us to apply the
results from the simulations of the symmetric $3$-level system to
an arbitrary $n$-level quantum system.

\end{document}